\documentclass{article}

\usepackage{arxiv}
\usepackage{graphicx}
\usepackage[utf8]{inputenc}
\usepackage[english]{babel}
\usepackage{amsmath}
\usepackage{amsfonts}
\usepackage{mathrsfs}
\usepackage{graphicx}
\usepackage{caption}
\usepackage{subcaption}
\graphicspath{ {./Figures/} }

\usepackage{algorithm}
\usepackage{algorithmic}
\AtBeginDocument{%
  \providecommand\BibTeX{{%
    \normalfont B\kern-0.5em{\scshape i\kern-0.25em b}\kern-0.8em\TeX}}}

\usepackage{mathrsfs} 

\usepackage{amsthm}
\newtheorem{theorem}{Theorem}
\newtheorem{lemma}[theorem]{Lemma}
\newtheorem{definition}{Definition}
\newtheorem*{remark}{Remark}

\AtBeginDocument{%
  \providecommand\BibTeX{{%
    \normalfont B\kern-0.5em{\scshape i\kern-0.25em b}\kern-0.8em\TeX}}}

\begin{document}
\title{A Case for Maximal Leakage as a Side Channel Leakage Metric}
\author{Benjamin Wu\\
Electrical and Computer Engineering\\
Cornell University\\
Ithaca, NY 14850\\
\texttt{bhw49@cornell.edu}\\
\And
Aaron B. Wagner\\
Electrical and Computer Engineering\\
Cornell University\\
Ithaca, NY 14850\\
\texttt{wagner@cornell.edu}\\
\And
G. Edward Suh\\
Electrical and Computer Engineering\\
Cornell University\\
Ithaca, NY 14850\\
\texttt{suh@ece.cornell.edu}\\
}

\maketitle              % typeset the header of the contribution
\begin{abstract}
Side channels represent a broad class of security vulnerabilities that have been demonstrated 
to exist in many applications. Because completely eliminating side channels often leads to 
prohibitively high overhead, there is a need for a principled trade-off between cost and leakage. 
In this paper, we make a case for the use of maximal leakage to analyze such trade-offs. 
Maximal leakage is an operationally interpretable leakage metric designed for side channels. 
We present the most useful theoretical properties of maximal leakage from previous work and 
demonstrate empirically that conventional metrics such as mutual information and channel capacity 
underestimate the threat posed by side channels whereas maximal leakage does not. We also study the cost-leakage trade-off as an optimization problem using maximal leakage. We demonstrate that not only can this problem be 
represented as a linear program, but also that optimal protection can be achieved using a combination of at most two deterministic schemes.
\end{abstract}

\section{Introduction}
Side channels represent a broad class of security vulnerabilities that have received significant 
attention from the cybersecurity community, especially after multiple side channel-based attacks 
have been demonstrated in recent years \cite{Kocher,wright_spot_2008,yan_study_2015}. 
Unfortunately, completely eliminating side channels can incur significant overhead, so practical 
protection techniques often aim to reduce information leakage as much as possible 
while maintaining acceptable performance. In this work, we aim to enable principled protection 
of side channels with passive adversaries. 

The foremost challenge in developing principled protection schemes for
side channels is quantifying the amount of leakage that occurs.
While various reasonable metrics have been considered in the 
literature~\cite{smith_foundations_2009,wagner_technical_2018}, these arguably do not capture, or even necessarily
upper bound, the utility of a given side channel to an attacker.
Recently, \emph{maximal leakage}~\cite{issa} was introduced 
as a metric that quantifies, in an operationally interpretable way,
how useful a given side channel is to an attacker. Armed with
such a metric, the system designer can rigorously trade off
security and performance when designing side-channel mitigation
schemes. Previous work \cite{issa_maximal_2016,issa_operational_2017,liao_tunable_2018} 
has shown that maximal leakage is computable in semi-closed form, 
that it provides a direct measure of how likely a side-channel attack will succeed, 
and that the metric is robust to changes in its underlying assumptions.
So far though, work on applying maximal leakage to practical side channels or on designing optimal protection schemes is limited. Our contributions in this paper are as follows: 
\begin{enumerate}
	\item We conduct an empirical study and find that information theory metrics such as mutual information and channel capacity typically underestimate the utility gained by the adversary whereas maximal leakage at least provides a upper bound. Moreover, we find that mutual information and channel capacity are metrics that result in sub-optimal protection when used to analyze side channel leakage.
	\item We also find that, despite its pessimistic formulation, maximal leakage does not consistently overestimate leakage. In many practical scenarios, it actually measures the adversary's utility without overestimating the threat.
	\item Under assumptions that encompass most timing and power side channels, we prove that cost-leakage optimality is achievable by close-to-deterministic protection, a result that is contrary to normal intuition.
	\item Finally, in case the LP is too large to solve (owing to large dimensionality), we propose a heuristic algorithm that approximates the aforementioned LP and bound the error incurred by using this algorithm.
\end{enumerate} 

In the first part of this paper, we study how maximal leakage can be used to provide strong 
side-channel protection guarantees in practical systems. We first discuss our assumed threat model and present a mathematical interpretation of side channels in Section 2. We then reiterate, from previous work, 
maximal leakage's theoretical advantages over traditional information theoretic metrics in Section 3. 
In Section 4, we provide an empirical study to demonstrate that these advantages are relevant to practical side channels. Our empirical studies lead us to further conclude that 
conventional information-theoretic leakage metrics such as mutual information and channel capacity result in suboptimal protection strategies. 
Surprisingly, we also find that mutual information and channel capacity, while appropriate measures of covert channel leaks \cite{millen_covert_1987}, 
actually underestimate the true threat posed by side channels.

In the second part of the paper, we consider how to design optimal
protection schemes that minimize the performance overheads over a side
channel given a target bound on maximal leakage.
We show that finding the optimal protection scheme can be formulated as a linear program (LP) and present a structural result on optimal protection schemes that allows complete computation of 
the entire optimal trade-off curve by solving a relatively small number of LPs in Section 5. 
This result applies under a set of broad conditions that include (but are not limited to) 
most timing, power, and compression based side channels. In addition, our structural result implies 
that cost-leakage optimality over maximal leakage is achieved by deterministic protection schemes with low implementation overheads. Here we also demonstrate that optimizing over mutual information and channel capacity result in suboptimal protection.
In Section 6, we present a heuristic that leverages knowledge of this structural result to rapidly 
approximate the full trade-off curve, at the cost of a small, bounded amount of sub-optimality. 
Together, these results enable fast computation of cost-leakage trade-offs with principled 
guarantees on leakage. Finally, we discuss related work in Section 7.

\section{Threat Model and Mathematical Interpretation}
In this section, we lay the groundwork to build up a rigorous model of side channels with passive adversaries. We first define our threat model and provide some examples. Then, under that threat model we formulate a mathematical interpretation of side channels and protection schemes.

\subsection{Threat Model} 
Our threat model is given as follows.
\begin{itemize}
	\item There exists a secret that the adversary aims to guess.
	\item An \emph{intermediate} is generated using a fixed function of the secret's value that may be either stochastic or deterministic. The distribution of the intermediate is known, but the conditional distribution of the intermediate given the secret is only necessarily known by the adversary. The intermediate value is not directly visible to the adversary.
	\item There exists a \emph{protection scheme}, which is either a stochastic or deterministic mapping of the intermediate to a side channel \emph{output}. This mapping from intermediate to outputs is presumed to be memoryless (it does not rely on past values of the intermediate) and is known by the adversary. The adversary can see the output value.
	\item The adversary must guess the secret, given the side channel output. They guess the secret value using the strategy that maximizes the probability of a correct guess after seeing the ouput.
	\item The adversary does not have any control over the victim's behavior that would affect the value of the secret, intermediate or the side channel output. We refer to such an adversary as a \emph{passive} adversary.
	\item The system designer only chooses the protection scheme. We presume that the joint distribution of the secret and output are fixed, once chosen. 
	\item We disregard system noise or measurement noise, but note that any independent randomness added to the side channel output cannot increase leakage any further.
\end{itemize}

To put these concepts into perspective, we give two relevant example side channels. First, consider RSA, an asymmetric key cryptosystem often used for key exchanges or identity validation. In systems that provide RSA decryption as a service, the decryption process may form a timing side channel if the decryption time varies with the value of the private key. Suppose for now that no protection scheme has been implemented. In this case, the victim is the implementation of RSA, the secret is the victim's private key, and the side channel output observed by the adversary is the runtime of the decryption algorithm (and the intermediate is the same as the output).

Second, consider speech coding. This is a type of data compression typically used in real-time services such as Voice-over-IP. Given a speech waveform, the encoder converts short snippets of the waveform into individual packets that can be later decoded to recover the snippet. Due to fidelity and rate constraints, it is common for the instantaneous data rate to vary as a function of the speech waveform. Again, suppose no protection scheme is implemented. In this case, the victim is the compression service and the side channel output is the packet size. In this case, the identity of the secret is not necessarily clear; it may be the uncompressed waveform, the transcript, the speaker ID, or even the language spoken. We will return to these two examples later on in Sections 4, 5, and 6.

\subsection{Mathematical Interpretation of Side Channels and Protection Schemes}
We represent the components of a side channel using random variables. First, the victim's secret is denoted as the random variable $U$, the intermediate is denoted as the random variable $X$, and the side channel output is denoted as the random variable $Y$. $X$ is a stochastic or deterministic function of $U$, $X(U)$, and $Y$ is a stochastic or deterministic function of $X$, $Y(X)$. We further assume that $X$ and $Y$ are discrete random variables with finite alphabets $\mathcal{X}$ and $\mathcal{Y}$, respectively, since truly continuous-valued variables are rare in computer systems. By construction, $U$, $X$, and $Y$ form a Markov Chain (denoted as $U-X-Y$), so $Y$ and $U$ are conditionally independent given $X$.

We assume that the relationship between $U$ and $X$ is immutable but that the system designer can control how $Y$ is generated from $X$. Note that introducing the intermediate $X$ into this side channel model does not lose us any generality. $X$ merely represents the boundary of our design problem: how $U$ maps to $X$ is assumed to be outside of our control while how $X$ maps to $Y$ is assumed to comprise our design problem. This model therefore includes side channels in which $U=X$ or $Y=X$ (although the latter simply corresponds to no protection). Therefore, the space of all possible choices of $Y(X)$ constitutes the space of all protection schemes.

Finally, we restrict our attention to $Y(X)$ that do not depend on past values of $X$ (as stated previously, memoryless protection). Such protection schemes ensure that $U$-values are only leaked by their corresponding $Y$-values, if at all. Furthermore, memoryless protection schemes can be represented simply as a transition matrix. This matrix is formatted such that each row corresponds to one element of $\mathcal{X}$, each column corresponds to one element of $\mathcal{Y}$, and each entry is the probability that an $X$-value is mapped to a particular $Y$-value. Figure~\ref{threat_model} summarizes the variables and assumptions we have made thus far. Next, we define some notation that makes use of the above and will be necessary as we discuss maximal leakage as a leakage metric.

 \begin{figure} \label{threat_model}
 	\begin{center}
 		\includegraphics{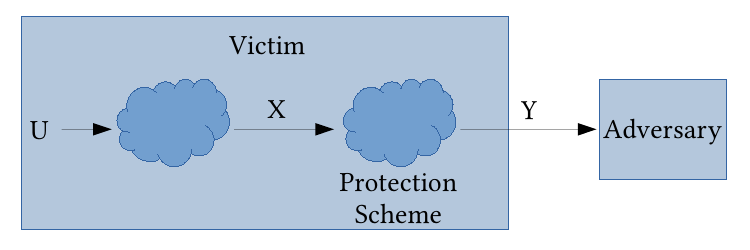} 
 		\caption{Summary of mathematical model of side channels.}
 		\label{threat_model}
 	\end{center}
 \end{figure}

\begin{definition} \label{variables}
	\textbf{(Basic Notation)}
	For random variables $X$ and $Y$, with alphabet sizes $|\mathcal{X}|=M$ and $|\mathcal{Y}|=N$, we define the following:
	\begin{itemize}
		\item $c(x,y)$ is the nonnegative (but not necessarily finite) cost of mapping each $x\in\mathcal{X}$ to each $y\in\mathcal{Y}$. Infinite cost entries correspond to illegal mappings. We refer to this function as the \emph{cost function} and the corresponding matrix $\{c_{xy}\}$ as the \emph{cost matrix}.	
		\item $p(x)$ is the distribution of $X$. $p(y)$ is the distribution of $Y$. $p(x,y)$ is the joint distribution of $X$ and $Y$. $p(y|x)$ is the conditional distribution of $Y$ given $X$.
		\item $\mathscr{C}(\textbf{A})=\mathscr{C}\{a_{xy}\} = \sum_{x,y}{p(x)c(x,y)a_{xy}}$ is the \emph{total cost} for any matrix $\textbf{A}=\{a_{xy}\}$. We will use $\mathscr{C}$ as a shorthand for this quantity, when the parameter matrix is implied.
		\item $L(\textbf{A})$ is the generic leakage value associated with any matrix $\textbf{A}=\{a_{xy}\}$ using some pre-specified leakage metric. When necessary, we will distinguish different metrics using subscripts (e.g. $L_{ML}(\textbf{A})$).
		\item $\textbf{P}=\{p_{xy}\}$ is an $M\times N$ transition matrix such that $p_{xy}=Pr(Y=y|X=x)\quad\forall x,y$ and $\mathscr{C}(\textbf{P})$ is finite. We refer to any matrix of this form as a \emph{protection scheme}. It is subject to typical transition matrix constraints (rows sum to 1, non-negative entries).
	\end{itemize}
\end{definition}

 %1. Introduce U,X,Y variables 2. Assumptions about U,X,Y
\section{Maximal Leakage}
While many leakage metrics have been proposed, we believe maximal leakage \cite{issa} is well-suited for analyzing side channel leaks. In this section we first define \emph{multiplicative gain leakage}, or \emph{mult-leakage}, a natural metric of side channel leakage. Then, we will summarize maximal leakage's relationship with mult-leakage, why it is necessary to use maximal leakage instead of mult-leakage, and the advantages of maximal leakage. Finally, we briefly discuss mutual information and channel capacity, how they relate to mult-leakage, and how they can be interpreted in the context of side channels.

\subsection{Multiplicative Gain Leakage}
Recall that, in any side channel, the adversary is interested in making an \emph{informed} guess of $U$ after observing $Y$. However, even without the side channel the adversary can make a \emph{blind} guess of $U$, which represents the case where $Y$ is independent of $U$. It stands to reason that a good leakage metric should reflect how the success rate of an informed guess compares to that of a blind guess. Thus, the ratio of the two is a natural measurement of leakage caused by the side channel. So, we define mult-leakage as:

\begin{equation}
L_{mult}(\textbf{P}) = \log{{\frac{\max_{\tilde{u}(\cdot)}{P(U=\tilde{u}(Y))}}{\max_{\tilde{u}}{P(U=\tilde{u})}}}}
\end{equation}

\noindent where the notation $\tilde{u}$ is a blind guess of $U$, and $\tilde{u}(Y)$ is a informed guess of $U$ after observing $Y$. Then, mult-leakage tells us the multiplicative gain on the adversary's probability of correctly guessing $U$ given the side channel (the logarithm is for scaling purposes). Alternatively, mult-leakage can be interpreted as the bit-difference of information between the informed and blind guesses. Thus, if the side channel leaks no information, then the best informed guess is no better than the best blind guess and the leakage is 0. On the other hand, if the side channel does leak information, mult-leakage tells us how much the adversary's guesses have been improved by the side channel and, ultimately, how useful the side channel is to the attacker. 

\subsection{Maximal Leakage}
Maximal leakage is defined as follows \cite{issa}:
\begin{align}
L_{ML}(\textbf{P}) = \max_{\substack{U:U-X-Y}}{\log{\frac{\max_{\tilde{u}(\cdot)}{P(U=\tilde{u}(Y))}}{\max_{\tilde{u}}{P(U=\tilde{u})}}}} = \max_{\substack{U:U-X-Y}}{L_{mult}{(\textbf{P}})}
\end{align}

By definition, maximal leakage is the worst-case mult-leakage over all possible $U$, which means it upper bounds mult-leakage. The reason we elect to use maximal leakage in favor of mult-leakage is that the latter is not always possible to compute, depending on whether the system designer knows $U$. Mult-leakage requires knowledge of $p(x)$, which is often difficult to characterize, especially for complex systems. Certainly, an approximation of mult-leakage is possible through data collection of $X$, but we are interested in upper-bounding leakage. 

At first blush, it seems that maximal leakage also requires knowledge of $p(x)$, but it turns out this is not the case. First, we note that maximal leakage is agnostic to $U$ by definition, so it is well-defined even in side channels where it isn't clear which secret the adversary is after. Furthermore, it can be shown that maximal leakage is can be computed as follows \cite{issa}:
\begin{equation}
L_{ML}(\textbf{P}) = \log{\sum_{y\in \mathcal{Y}}{\max_{x\in\mathcal{X}}{p(y|x)}}}
\end{equation}
In other words, maximal leakage requires knowledge of $p(x)$ only in terms of its support, the set of $X$-values with non-zero $p(x)$. And since we have restricted our attention to memoryless protection, we only need to find the maximum value in each column (ignoring $X$-values with probability of 0) of $\textbf{P}$, sum these values, and take the log of the sum to compute it. Moreover, the maximum mult-leakage is achieved by a particularly useful $U$, referred to as a \emph{shattering} $U$ by the authors. This type of $U$ is characterized by the following two properties. 

\begin{definition} \label{shatteringDef}
	\textbf{(Shattering U)} For any random variable $X$, $U$ is \emph{shattering} if:
	\begin{itemize}
		\item $U$ is uniformly distributed over a finite alphabet $\mathcal{U}$
		\item For each $u\in\mathcal{U}$, $x=X(u)$ is a deterministic value.
	\end{itemize}
\end{definition}

A conceptual example of a shattered distribution is shown in Figure \ref{shattering}. Here, each blue square corresponds to an equally-probable potential value of $U$ (34 distinct values total, in this example). An example distribution of $X$ is shown, where each possible $X$-value corresponds to some number of $U$-values. 

\begin{figure}
	\centering
	\includegraphics{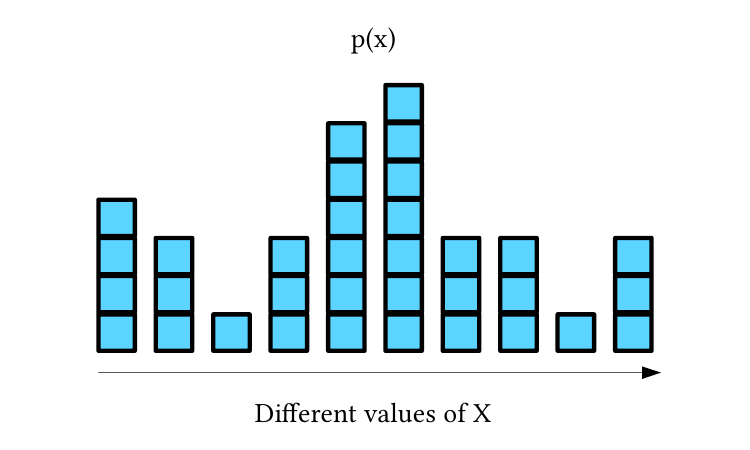}
	\caption{Conceptual example of shattering U.}
	\label{shattering}
\end{figure}

The significance of this shattering distribution is twofold. First, for any shattering $U$, maximal leakage equals mult-leakage. In such cases, maximal leakage measures the the utility gained by the adversary from the side channel under any protection scheme. Second, the shattering $U$ is quite representative of side channels in which $U$ is known to be an encryption key, if keys are selected uniform randomly from some (not necessarily known) space of allowed keys and the baseline side channel process $X(U)$ is not stochastic. For such channels, maximal leakage, though pessimistic by design, exactly captures the side channel leakage. 

\subsection{Other metrics}
\noindent\textbf{Mutual Information:} As one of the most basic information-theoretic metrics, mutual information has been used to evaluate side channel protection in some previous work on side channel protection \cite{gong_quantifying_2016,zhou_camouflage:_2017}.
	\begin{equation} \label{MIDef}
		L_{MI}(\textbf{P})= \sum_{x\in\mathcal{X},y\in\mathcal{Y}} p(x)p(y|x) \log{\frac{p(y|x)}{\sum_{x\in\mathcal{X}}{p(x)p(y|x)}}} 
\end{equation}	
	In plain terms, mutual information measures the amount of information shared between random variables $X$ and $Y$. However, mutual information  requires knowledge of $p(x)$ and is also is upper bounded by maximal leakage \cite{issa}. Since we know that mult-leakage equals maximal leakage for the shattering $U$, we note that mutual information may underestimate leakage in such cases (whether it does in practice is a question we will explore in Section 4).
	
\noindent\textbf{Channel Capacity:} Another basic information-theoretic metric, channel capacity, is used to measure the rate of reliable communication over a noisy channel.
	\begin{equation}
		L_{CC}(\textbf{P})= \max_{p(x)}L_{MI}(\textbf{P})
	\end{equation}
	As such, it is a useful metric for leakage in covert channels where a sender deliberately encodes messages to a receiver. At first blush, it may seem that channel capacity should bound the rate of information leakage in a side channel since the adversary doesn't have the luxury of a cooperative sender. Interestingly, channel capacity is also upper bounded by maximal leakage \cite{issa}. This is because channel capacity assumes the sender and receiver are interested in complete, reliable decoding of messages passed over the channel \cite{issa}. In a side channel scenario, such an assumption is unnecessary because the adversary does not need to completely recover the secret to pose a threat. Indeed, many side channels do not risk complete recovery even without protection.\footnote{For a more nuanced discussion of this topic, refer to Section III of \cite{issa_operational_2017}}

\begin{figure}
	\centering
	\includegraphics{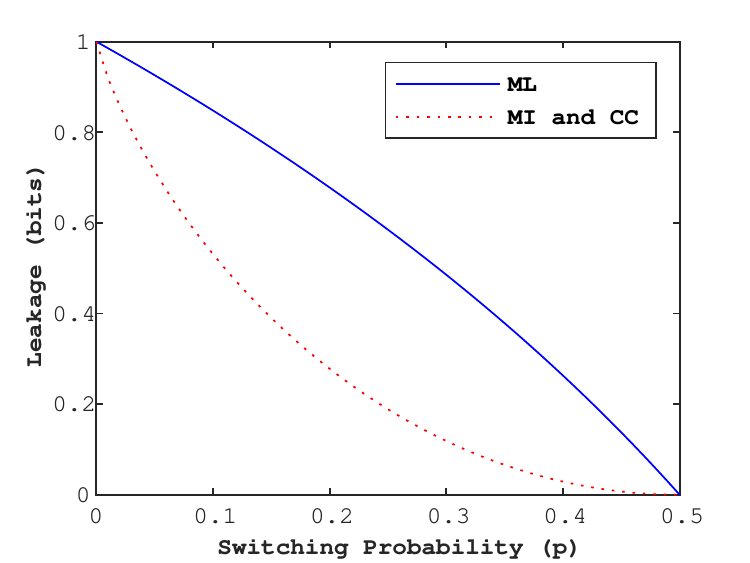}
	\caption{MI, CC, and ML on a BSC with uniform random $X$. Note a BSC is a channel where a uniformly-distributed, single bit input is mapped to a single bit output, but there exists a probability $0\leq p\leq 1$ that the input bit will be flipped. Due to symmetry, we may constrain our attention to $0\leq p\leq 0.5$.}
	\label{BSC_comp}
\end{figure}

Finally, to demonstrate that mutual information, channel capacity, and maximal leakage are numerically comparable (a fact that may not be immediately obvious from their definitions), we analytically compute all three metrics for a binary symmetric channel (BSC) with various switching probabilities $p$. This comparison can be seen in Figure~\ref{BSC_comp}. We can see that all three metrics agree on the worst and best case leakage values, when $p$ is 0 (all metrics agree leakage is 1) and whe $p$ is 0.5 (all metrics agree leakage is 0). However, both mutual information and channel capacity can be less than maximal leakage by an arbitrarily large ratio. Thus, an open question at this point is whether, in practice, maximal leakage is too conservative or whether mutual information and channel capacity are truly underestimating the leakage. We address this question next. %1. Desired properties of leakage metrics 2. Introduce ML definition and theorem 3. Explain MI and CC, BSC example 4. RSA experiments (MI and CC are not upper bounds) 5. Gaussian noise on Gaussian X (MI doesn't represent channel leakage well)
\section{Empirical Study of Leakage Metrics}
In the previous section, we showed that mutual information/channel capacity could underestimate mult-leakage, at least in theory. However, to truly motivate the use of maximal leakage over these alternatives requires that such gaps between mutual information/channel capacity and mult-leakage do exist in practice, which we will demonstrate in this section. 

The rest of the section will proceed as follows. First, we present an example side channel with which we will compare these metrics: an RSA decryption timing channel. Second, we demonstrate the existence of a significant gap between mutual information and mult-leakage under square-and-multiply implementations of RSA and under GNU's multiple precision (GMP) library implementation. We argue that the existence of such gaps indicates that mutual information and channel capacity are problematic when cost-leakage optimality is desired or when a real bound on leakage is needed. 

\subsection{RSA Decryption Timing Channel}
For the rest of this section, we consider a timing channel involving RSA as seen in Figure~\ref{RSA_threat_model}. In this side channel, Alice serves many clients who need to use Alice's public key to encrypt messages to her over a network. For each such encrypted message, Alice uses her private key to decrypt the message and then responds. The adversary, Eve, is not one of Alice's clients but is capable of observing the network traffic coming to and from Alice, perhaps by employing a packet sniffer (hence, Eve is a passive adversary). Eve sees when each decryption request arrives and when Alice responds from the sequence and timing of packets, which she uses to guess Alice's private key. Here, $U$ is the private key, $X$ is the decryption time, and $Y$ is the length of time between when Alice receives each message and when she sends a response to the client. To implement a protection scheme, we must choose $Y(X)$, which is how long to delay Alice's response on top of the true decryption time.

Finally, we assume that Alice's private key was chosen uniformly at random from all binary strings of a fixed length (as opposed to only legal keys based on the RSA cryptosystem). This is a choice of convenience to facilitate our experiments, but doing so does not affect our conclusions, as we will explain later in this section.

\begin{figure}
	\centering
	\includegraphics{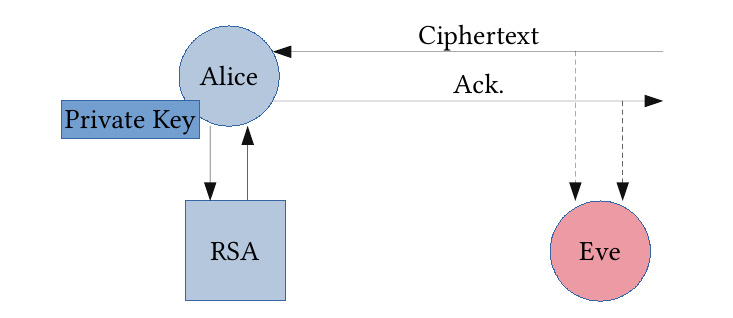}
	\caption{RSA timing side channel example.}
	\label{RSA_threat_model}
\end{figure}

\subsection{Square-and-Multiply Implementation of RSA}
We first consider the square-and-multiply implementation of modular exponentiation, the main sensitive operation of RSA decryptions. Pseudocode for the square-and-multiply implementation of RSA decryption is as follows:

\begin{algorithmic}[1]
	\STATE Inputs $c$ (ciphertext), $u$ (private key), $n$ (modulus)
	\STATE $r \gets 1$
	\STATE $c\gets c\mod n$
	\WHILE {$u>0$}
		\IF {$u\mod2 == 1$} 
			\STATE $r\gets (m*c)\mod n$
		\ENDIF
		\STATE $u\gets u>>1$
		\STATE $c\gets (c*c)\mod n$
	\ENDWHILE
	\RETURN $r$
\end{algorithmic}

The timing channel arises from the if statement in line 5. The modular multiplication of the result and ciphertext only occurs if the next bit of the private key is 1. From this fact, the adversary can deduce the weight (the number of 1s) of the private key. We make several simplifying assumptions and define the parameters of this experiment as follows:

\begin{itemize}
	\item Ignore confounding factors, such as system noise or network delay. $Y$ is simply equal to $X$ plus any delay we choose to add.
	\item Assume all 1024-bit sequences are valid keys.
	\item Assume the bits of the key are independently and identically distributed Bernoulli random variables. This results in a uniform-randomly selected key out of all 1024-bit binary strings.
	\item Let $\textbf{U}=[U_1, U_2, ...U_{1024}]$ be a random vector representing the value of the private key.
	\item Let $X=\sum_{i=1}^{1024} U_i$ be a random variable representing the weight of the private key.
	\item Assume that the multiplication in line 6 of the above pseudocode takes a fixed $K$ milliseconds to execute each time it is called.
	\item Let $Z$ be a binomial random variable with fixed probability $p=\frac{1}{2}$ and size parameter $m$ (which we vary).
	\item Let $Y=X+Z$. In other words, our protection scheme is independently added binomial noise.
	\item Let $c(x,y)=\begin{cases}
		y-x \quad \text{if } y\geq x\\
		\infty\quad \text{otherwise}
	\end{cases}$
	
	Note that this cost matrix enforces that any protection with finite total cost must be upper triangular. Moreover, note that with this cost matrix and independent binomial delays, the total cost is $\frac{m}{2}$.
\end{itemize}

Here, we can analytically compute mutual information, channel capacity, and maximal leakage for many different values of $m$ (the size parameter of the binomial-distributed noise) and plot them on the same axes, as seen in Figure~\ref{SAM_Impl}. We note that there exists a large gap between mutual information and maximal leakage that sharply shrinks as we approach no protection. A sizeable gap exists between maximal leakage and channel capacity as well. Recall that, since $U$ is shattering, maximal leakage equals mult-leakage. Conversely, this implies that both mutual information and channel capacity underestimate mult-leakage in this example, and thus also underestimate the adversary's utility. Note that, had we chosen a key uniform-randomly from the set of all feasible keys (according to the RSA cryptosystem), $U$ would still have been shattering.

\begin{figure} 
	\centering
	\includegraphics{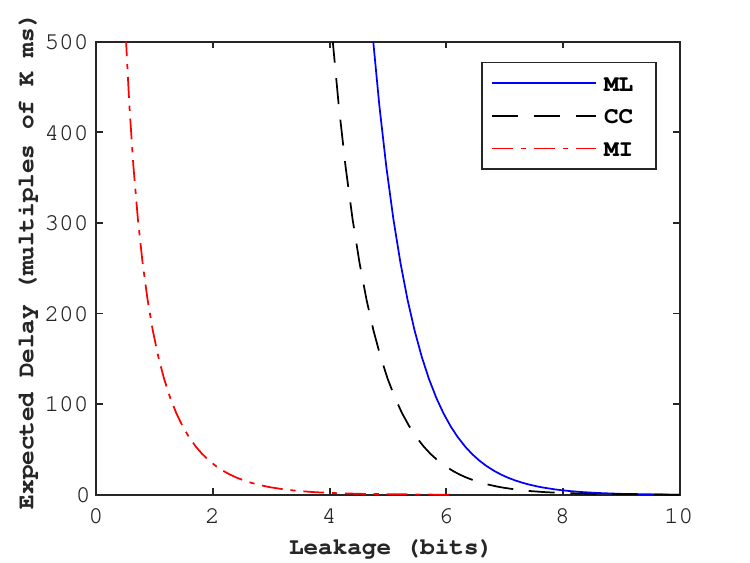}
	\caption{Metric comparison for square-and-multiply RSA with independent binomial noise on a 1024-bit key decryption. Legend: ML is maximal leakage, CC is channel capacity, and MI is mutual information. Note that ML equals mult-leakage since the key is a shattering $U$.}
	\label{SAM_Impl}		
\end{figure}

\subsection{GMP Implementation of RSA}
Here, we consider an implementation of modular exponentiation where the weight of the key isn't directly leaked but instead the decryption time varies with the key in some other way. We show that maximal leakage's gaps with mutual information and channel capacity are still significant even in this case. We use GNU's multiple-precision (GMP) library's implementation of modular exponentiation. Here, we perform essentially the same experiment as before. The assumptions and parameters of the experiment are as follows (only ones that are different from the square-and-multiply implementation will be listed):

\begin{itemize}
	\item Let $\textbf{U}=[U_1, U_2, ...U_{16}]$ be a random vector representing the value of the private key. Note that we are using a 16-bit key here so that it is possible to exhaustively collect decryption timing data for all private keys and ciphertexts, to compute the distribution of decryption times.
	\item Let $X$ be the random variable representing the execution time (in cycles) of GMP's modular exponentiation on an Intel i7 core. The decryption time of each private key varies with the ciphertext, so we uniform randomly selected a fixed ciphertext for the purposes of this experiment. The distribution of $X$ we used can be seen in Figure~\ref{histogram}.
	\item Choose the alphabet of $Y$, $\mathcal{Y}$, as follows. Choose a noise \emph{width} $w$. Extend $\mathcal{X}$ by $w$ elements, each spaced by the most common difference between consecutive elements in $\mathcal{X}$ (in case of a tie, choose the smallest common difference). So for example, suppose $\mathcal{X}=\{1,\ 5,\ 7,\ 9,\ 11,\ 13\}$ and $w=4$. Then $\mathcal{Y}=\{1,\ 5,\ 7,\ 9,\ 11,\ 13,\ 15,\ 17,\ 19,\ 21\}$ since the most common interval between consecutive elements in $\mathcal{X}$ is 2.
	\item Let $Z$ be a binomial random variable with fixed probability $p=\frac{1}{2}$ and size parameter $w$ (which we vary). 
	\item Let $Y(X)$ be  defined as follows. Given $x\in\mathcal{X}$, generate a $Z$-value $z$. Note that $z$ is an integer; choose the $z$th larger element than $x$ in $\mathcal{Y}$.
\end{itemize}

As before, we can directly compute mutual information, channel capacity, and maximal leakage for many different values of $w$. Plotting mutual information, channel capacity, and maximal leakage against total cost (Figure~\ref{GMP_Impl}), we find that a gap exists between maximal leakage and mutual information/channel capacity. These results reaffirm our earlier observation that mutual information and channel capacity underestimate the advantage given to the adversary in practice.

\begin{figure}
	\centering
	\begin{subfigure}[t]{0.45\columnwidth}
		\includegraphics[width=\columnwidth]{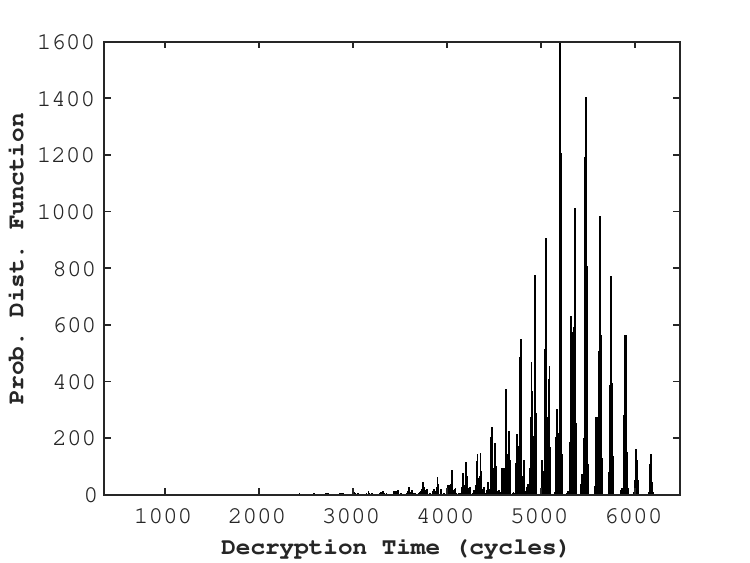}
		\caption{}
		\label{histogram}
	\end{subfigure}
	\begin{subfigure}[t]{0.45\columnwidth}
		\includegraphics[width=\columnwidth]{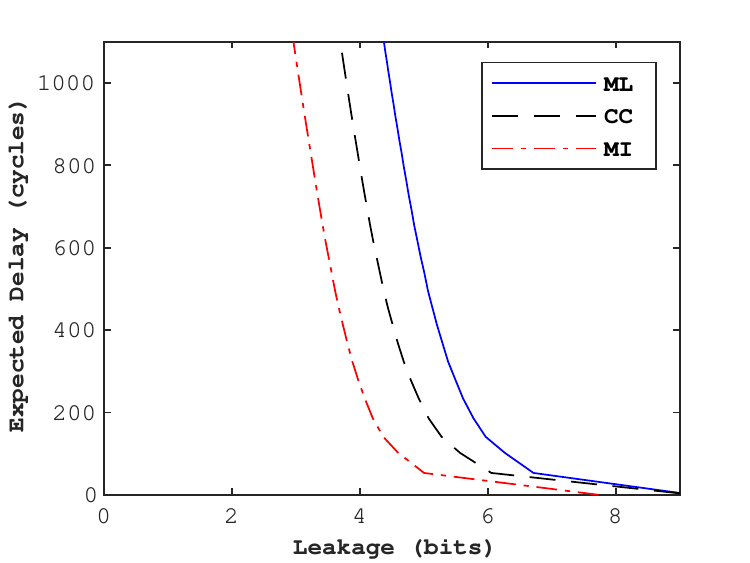}
		\caption{}
		\label{GMP_Impl}
	\end{subfigure}
	\caption{(a) Empirical distribution of decryption time $X$ for uniform randomly selected ciphertext $1110011001001100_2$ ($58956_{10}$). (b) Metric comparison for GMP implementation of RSA with binomial random extension on 16-bit key decryption. Again, note that ML equals mult-leakage.}
\end{figure}

Here, we remark that both in the case of the GMP implementation and in the earlier square-and-multiply implementation, we obtained trade-off curves for mutual information, channel capacity, and maximal leakage with very similar shapes. So, it may be tempting to suggest that there is little difference in usage between the three metrics, since their trade-off curves are so similar in shape. However, there are two factors to keep in mind. First, in the above experiments we have only used independent (of $X$) random padding. Second, we compared how different metrics behave for the same protection scheme. Essentially, we have not shown that mutual information, channel capacity, and maximal leakage agree on a relative ordering of how secure an arbitrary pair of protection schemes (possibly not independent of $X$). Indeed, in the next section, we will prove that mutual information/channel capacity disagree qualitatively on what kinds of protection are optimal. We will even show an example of two protection schemes that maximal leakage disagrees with mutual information/channel capacity in terms of the relative ordering of leakage. 
\section{Optimal ML-Cost Trade-offs}
Using maximal leakage as our leakage metric, we will show two key facts about the minimization of total cost subject to an upper bound on maximal leakage, or the optimization over maximal leakage. First, we demonstrate that the optimization over maximal leakage can be written as a linear program. While this fact is relatively simple to verify, its significance lies in that it greatly simplifies the process of solving the optimization itself (which, in general, is not a trivial feat). Second, as the main theorem of this section, we prove that under certain constraints on the cost matrix $c(x,y)$ that are quite common among side channels, optimality under maximal leakage can be achieved with easy-to-implement deterministic protection schemes. The remainder of the section is dedicated to stating these results rigorously, explaining their implications, and finally comparing maximal leakage optimal schemes with mutual information and channel capacity ones.

\subsection{Formulating the Optimization}
\begin{definition} \label{defTerms}
\textbf{(Definitions)} 

The following definitions are needed to formulate the optimization. For random variables $X$ and $Y$, with alphabet sizes $|\mathcal{X}|=M$ and $|\mathcal{Y}|=N$, we define the following:
	\begin{itemize}
		\item We retain the defined variables and functions given in Definition~\ref{variables}, but will at this point we will retire the previous notation to distinguish between various leakage metrics (e.g. $L_{ML}(\textbf{P})$) in favor of the next item.
		\item $\mathscr{L}(\textbf{A})=\mathscr{L}(\{a_{xy}\}) = \sum_y{\max_x{a_{xy}}}$ is the \emph{exponentiated maximal leakage} (or \emph{exp-leak}, for short) of any matrix $\textbf{A}=\{a_{xy}\}$. We will use $\mathscr{L}$ as a shorthand for this quantity, when the parameter matrix is implied. Note that minimizing over exp-leak is equivalent to minimizing over maximal leakage.
		\item a given protection scheme \textbf{P} is \emph{deterministic} if all $p_{xy}$ equal 0 or 1. It is \emph{stochastic} otherwise.
		\item an $(L,C)$ pair is \emph{achieved} by \textbf{P} if $\mathscr{L}(\textbf{P}) \leq L$ and $\mathscr{C}(\textbf{P}) \leq C$.
		\item an $(L,C)$ pair is \emph{achievable} if there exists such a \textbf{P} that $(L,C)$ is achieved by it.
		\item the set $S$ is the set of all achievable $(L,C)$ pairs.
		\item $C^*(L) = \inf{[C:(L,C)\in S]}$. We refer to $C^*(L)$ evaluated for all values of $L$ as the tradeoff curve and the set of points $S_b = [(L,C)\in S|C = C^*(L)]$ as the boundary of $S$.
		\item \textbf{P} is \emph{optimizing in $S$} if $\mathscr{C}(\textbf{P}) = C^*(\mathscr{L}(\textbf{P}))$ (i.e. if \textbf{P} achieves a point on the boundary of $S$).		
		\item the set $S_d$ is the set of all points in $S$ that can be achieved by a deterministic protection scheme.
		\item an $(L,C)$ pair is \emph{achievable in $S_d$} if there exists a deterministic protection scheme \textbf{P} that achieves $(L,C)$.
		\item $C^*_d(L) = \inf{[C:(L,C)\in S_d]}$.
		\item \textbf{P} is \emph{optimizing in $S_d$} if $\mathscr{C}(\textbf{P}) = C^*_d(\mathscr{L}(\textbf{P}))$. Note that a \textbf{P} that is optimizing in $S_d$ is not necessarily a deterministic protection scheme.
	\end{itemize}
\end{definition}

%%%%%%%%%%%%%%%%%%%%%%%%%%%%%%ENUMERATE SYMBOLS%%%%%%%%%%%%%%%%%%%%%%%%%%%%%%
\begin{remark}
\textbf{(Set Indexing)} 
We will choose to let $\mathcal{X}= \{x_1,x_2,...x_M\}$ and $\mathcal{Y}=\{y_1, y_2,...y_N\}$.
\end{remark}

Here, we consider the optimization problem over maximal leakage:
\begin{equation}
\begin{split}
	C^*(L) = \underset{\textbf{P}}{\min{}} \mathscr{C}(\textbf{P}) \quad \text{s.t. } &\mathscr{L}(\textbf{P})\leq L, \  \sum_y{p_{xy}} = 1 \ \forall x,\\  
	&p_{xy}\geq 0 \ \forall\ x,y\\
\end{split}
\label{optProbML}
\end{equation}

This can be rewritten as an LP as follows:
\begin{equation}
\begin{split}
	C^*(L) = \underset{p_{xy},q_y}{\min{}} \mathscr{C}(\textbf{P}) \quad \text{s.t. } &\sum_y{q_y}\leq L, \  \sum_y{p_{xy}} = 1 \ \forall x,\\
	&p_{xy}\geq 0, \ p_{xy}\leq q_y, \ \forall\ x,y\\
\end{split}
\end{equation}

%Unfortunately, simply knowing that the optimization is an LP is not sufficient in many practical scenarios. Real side channels may potentially have very large alphabets $\mathcal{X}$ and $\mathcal{Y}$, meaning that solving the optimizations may be time-consuming. Ideally, we would like to minimize the number of optimizations needed to compute the full trade-off curve. Moreover, once an appropriate protection scheme is selected, the size of the transition matrix impacts the feasibility of implementation. With these two computational constraints in mind, our next subsection is dedicated to our key structural result, which helps address these problems.

\subsection{Structural Result and Proof}
%%%%%%%%%%%%%%%%%%%%%%%%%%%%%%%%%%CONVEXITY%%%%%%%%%%%%%%%%%%%%%%%%%%%%%%%%%%
\begin{remark} \label{ConvexityThm}
\textbf{(Convexity)} 
$C^*(L)$ is a convex function of $L$. The proof follows from standard arguments. The convexity of the optimal trade-off curve is significant in that it allows for a useful qualitative assessment of the optimization problem. That is, adding a little protection on top of an unprotected side channel is very costly, but relaxing a zero-leakage scheme buys more cost reduction.
\end{remark}

%%%%%%%%%%%%%%%%%%%%%%%%%%%%%%%NON-CONVEXITY OF C^*_d%%%%%%%%%%%%%%%%%%%%%%%%
\begin{remark}
Note that $C^*_d(L)$ is not convex since a deterministic protection schemes necessarily has an integer exp-leak value. The space $S_d$ is a subset of $S$ given by all $(L,C)\in S$ pairs dominated by the set of finite points $(L,C^*_d(L))$ for integer $L$ values . The boundary of $S_d$ is shaped like a descending staircase, and $S_d$ is the set of all points above and to the right of this stair-like boundary.
\end{remark}

%%%%%%%%%%%%%%%%%%%%%%%%%%%%%%%staircase nondecreasing%%%%%%%%%%%%%%
\begin{definition}\label{stairNondec}
\textbf{(Cost Constraints)} 

We refer to a cost function/matrix that satisfies the following constraints as \emph{staircase nondecreasing}:
\begin{enumerate}
		\item For $0<i<j\leq M$ and all $y\in\mathcal{Y}$, if $c(x_i,y)=\infty$, then $c(x_j,y)=\infty$. (i.e. if one matrix element is infinite, then that column is infinite all the way down).
		\item For $0<i<j\leq N$ and all $x\in\mathcal{X}$, if $c(x,y_i)<\infty$, then $c(x,y_i)\leq c(x,y_j)<\infty$. (i.e. excluding infinities, each row of the matrix is nondecreasing from left to right).
	\end{enumerate}
\end{definition}

Note that staircase nondecreasing cost matrices are exemplified by upper triangular cost matrices (where all entries below the diagonal are infinite cost) with ordered cost entries for each row. This special case of staircase infinite, nondecreasing cost matrices is typical of most power and timing side channels, since one cannot map power consumption or latency to a value less than itself.

%%%%%%%%%%%%%%%%%%%%%%%%%%%%MAIN LEMMA%%%%%%%%%%%%%%%%%%%%%%%%%%%%%%%%%%%%%%%
%\begin{lemma} \label{keyLemma}
%\textbf{(Key Lemma)} 
%
%If $c(x,y)$ is staircase nondecreasing,
%
%$\underset{(\mathscr{L},\mathscr{C})\in S}{\min} \mathscr{C} + \alpha \mathscr{L} = \underset{(\mathscr{L},\mathscr{C})\in S_d}{\min} \mathscr{C} + \alpha \mathscr{L}\quad\forall\alpha>0$
%\end{lemma}
%
%The immediate consequence of Lemma \ref{keyLemma} is that the sets $S$ and $S_d$ have the same supporting hyperplanes for all values of $\alpha$. Thus, Lemma \ref{keyLemma} suggests that the convex hulls of $S$ and $S_d$ are identical. Since the convex hull of $S_d$ is piecewise linear, this lemma plays a major role in the proof of the following theorem. The full proof of this lemma can be found in Appendix \ref{appLem}.

%%%%%%%%%%%%%%%%%%%%%%%%%%%%%%%%%%%MAIN THEOREM%%%%%%%%%%%%%%%%%%%%%%%%%%%%%%
\begin{theorem} \label{mainThm}
\textbf{(Main Theorem)} 

If $c(x,y)$ is staircase nondecreasing, then
\begin{enumerate}
	\item $\underset{(\mathscr{L},\mathscr{C})\in S}{\min} \mathscr{C} + \alpha \mathscr{L} = \underset{(\mathscr{L},\mathscr{C})\in S_d}{\min} \mathscr{C} + \alpha \mathscr{L}\quad\forall\alpha>0$
	\item For all $L\geq 1$, $(L,C^*(L))$ can be achieved by $\textbf{P}=\lambda\textbf{P}_1+(1-\lambda)\textbf{P}_2$ for some $\lambda\in[0,1]$ and some deterministic protection schemes $\textbf{P}_1$ and $\textbf{P}_2$, such that $\mathscr{L}(\textbf{P})\leq L$ and $C^*(L) \leq \lambda C^*_d(\mathscr{L}(\textbf{P}_1))+(1-\lambda)C^*_d(\mathscr{L}(\textbf{P}_2))$.
\end{enumerate}
\end{theorem}

The proof proceeds by taking an optimizing stochastic protection scheme and showing that it can be made more deterministic without losing optimality. Full proof in Appendix~\ref{appProofs}.

\subsection{Implications of Theorem \ref{mainThm}}
The main implication of Theorem~\ref{mainThm} is that deterministic protection schemes are sufficient to achieve optimality over maximal leakage. The first part of the theorem essentially states that the supporting hyperplanes of the trade-off space are achieved by deterministic schemes. The second part follows from the first and states that the constrained optimization is solved by at least the next best thing: a mixture of at most two deterministic schemes.

First, recall the implementation benefits of deterministic protection schemes. Regardless of the specific cost function (as long as it is staircase nondecreasing) or application, any deterministic protection scheme can be compressed to an $N\times 2$ matrix (or smaller) recording which $Y$-value each $X$-value maps to. In addition, deterministic schemes are resistant to averaging attacks, where the adversary attempts to learn additional information by gathering statistics of $Y$, since the same $X$ value always maps the the same $Y$ value. In the event that a mixture of two deterministic schemes is needed, one may implement a pre-determined schedule alternating between the two deterministic schemes for each $X\rightarrow Y$ mapping. Here, while the leakage of individual observations of $Y$ will change over time, we can enforce the desired long-run bound.

Second, the proof of the main theorem induces an algorithm by which one may take any optimizing protection scheme and convert it to a deterministic form, so that the discussed benefits can be leveraged. This algorithm simply performs the procedures specified in Definitions~\ref{waterFillingDef} and~\ref{QGen} in Appendix~\ref{appProofs} recursively.

Third, if it is necessary to solve the entire optimal trade-off curve (for example, if on-the-fly tuning of leakage is expected), Theorem~\ref{mainThm} states that it is only necessary to solve for integer exp-leak points and then connect the dots so that the overall curve is convex. Also, for small deviations in the leakage bound, tuning can be done simply by changing the mixture proportion.

\subsection{Comparing Optimal Protection Under Different Metrics}

\begin{figure}
	\centering
	\begin{subfigure}[t]{1.1\columnwidth}
		\includegraphics[width=\columnwidth]{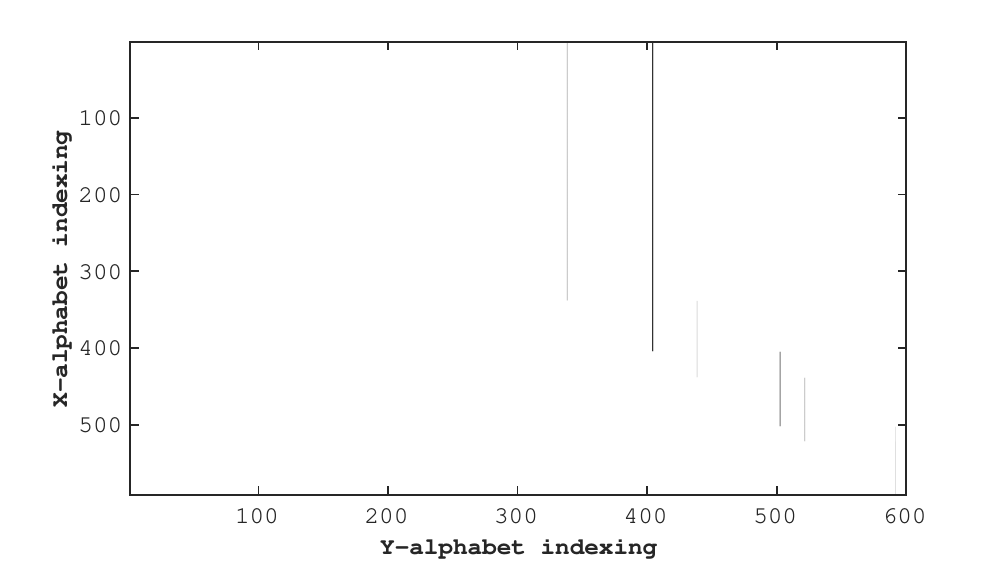}
		\caption{}
		\label{RSA_5percent}
	\end{subfigure}
	\begin{subfigure}[t]{1.1\columnwidth}
		\includegraphics[width=\columnwidth]{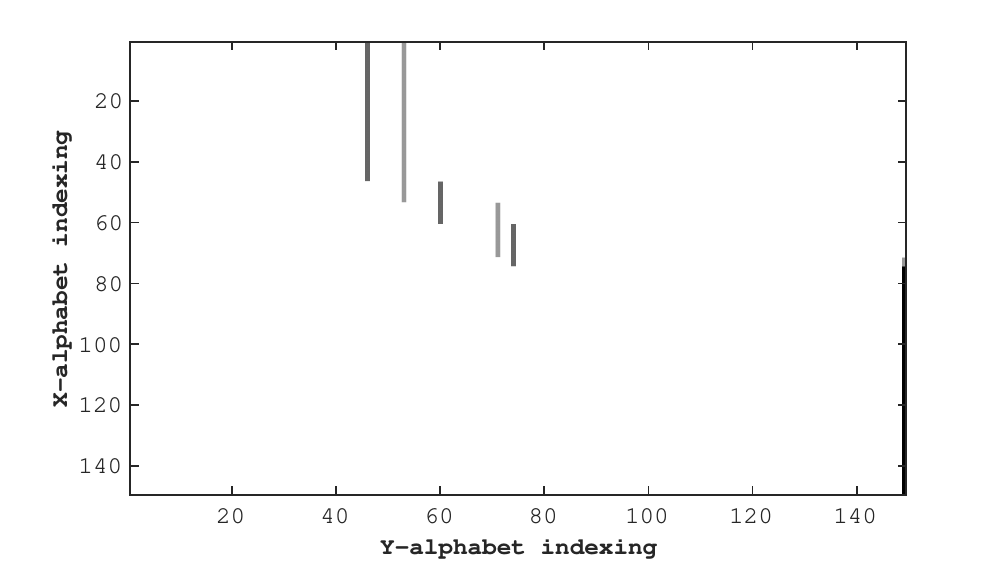}
		\caption{}
		\label{VoIP_20percent}
	\end{subfigure}
	\caption{(a) ML-optimal solution for 5\% delay overhead on RSA decryption times. ML=1.6862 bits. Note ML equals mult-leakage(b) ML-optimal solution for 20\% padding overhead on VoIP packet sizes. ML=1.8111 bits.}
\end{figure}

Finally, we will compare maximal leakage optimal schemes (\emph{ML-optimal} schemes) to mutual information and channel capacity optimal schemes (\emph{MI-optimal} and \emph{CC-optimal} schemes). First, we will simply show sample ML-optimal schemes for the RSA decryption side channel (the GMP implementation) and for a packet size side channel based on VoIP applications. Then, using these sample schemes, we will discuss the qualitative nature of MI-optimal and CC-optimal schemes.\newline

\noindent\textbf{ML-optimal protection for RSA decryption: } Using the GMP-based decryption timing data from Section 4, we use Gurobi (a convex optimization solver) to solve the inverted optimization (i.e. minimize leakage subject to a cost bound) from Equation \label{optProbML} for a total cost bound of 5\% delay overhead. Doing so achieves maximal leakage of 1.6862 bits. Note that in this type of side channel, one relevant class of deterministic protection schemes is emph{thresholding} where, for an ascending sequence of thresholds, all $X$ values less than or equal to the smallest threshold are mapped to the first threshold, all $X$ values greater than the first and less than or equal to the second threshold are mapped to the second, and so on. As it turns out, the resulting protection scheme from this experiment, is shown in Figure \ref{RSA_5percent}, and can be described as the combination of the thresholding schemes $\textbf{P}_1$ with thresholds $x_{338}, x_{438}, x_{521}, x_{591}$ and $\textbf{P}_2$ with thresholds $x_{404}, x_{502}, x_{591}$. This particular scheme can be implemented by a schedule that uses $\textbf{P}_1$ 21.8\% of the time and $\textbf{P}_2$ 78.2\% of the time.\newline

\noindent\textbf{ML-optimal protection for VoIP: } Here we present the packet size channel for speech coding in the context of Voice-over-IP (VoIP) applications and then perform an analogous experiment to assess the ML-optimal scheme. Previous work demonstrated that there exists a side channel leak through the sizes of packets sent over networks\cite{wright_spot_2008}. It has even been shown that such side channels allow packet sniffers to partially recover or reconstruct spoken phrases \cite{doychev_yes_2009,wright_uncovering_2010}.

The victim system is a typical VoIP application, which operates by encoding fixed-length time intervals (called a ``frame'') of sound waveforms into one packet per time interval. In particular, VoIP system designers favor a form of variable bitrate (VBR) compression, which reduces bandwidth usage and improves recovered speech quality. 

We assume that the adversary is interested in reconstructing the transcript, or the text of what was spoken. The adversary observes the final payload size of each packet, $Y$. $X$ is the un-padded packet size produced by the speech codec (a coder-decoder used to compress and decompress human speech). Our protection scheme maps $X$ to $Y$ by padding each packet independently.

For our experiments, we used Mozilla's CommonVoice\footnote{https://voice.mozilla.org/en} English dataset. For the speech codec, we used Opus, an efficient open-source codec endorsed by the IETF, set to 24 kbps VBR with a frame size of 20 ms. Under these settings, Opus encodes each frame to one of 151 different packet sizes. We encoded approximately 572 hours worth of human speech to obtain the distribution $p(x)$. 

Then, we set the cost matrix to be the number of bytes of padding for each packet ($c(x,y)=y-x$ if $y>=x$ and $c(x,y)=\infty$ otherwise), and again computed the inverted optimization using Gurobi. Solving the inverted optimization for a 20\% padding overhead gives the solution seen in Figure~\ref{VoIP_20percent}, and it can be decomposed into two deterministic protection schemes in a similar fashion to the RSA optimization. \newline

\noindent\textbf{MI and CC optimal protections: } Now, we will show that deterministic schemes typically will not be either MI-optimal or CC-optimal, except in some edge cases such as the zero-leakage case. This observation will further lead us to conclude that, especially for side channels with shattering $U$, MI and CC optimal schemes will commonly result in suboptimal (in terms of the optimization over mult-leakage) protection.

First, recall the definition of mutual information as given in Equation~\ref{MIDef}. Since $p(x)$ is fixed for the optimization over mutual information, the only active variable in mutual information is $p(y|x)$. Since mutual information is nonlinear and convex over $p(y|x)$, we can expect that it will be optimized by protection schemes in the interior of the feasible set. In other words, we expect that MI-optimal schemes will tend to be stochastic. Indeed, a simple experiment confirms this. Using the the alphabets $\mathcal{X}=\{x_1, x_2, x_3, x_4\}$ and $\mathcal{Y}=\{y_1, y_2, y_3, y_4\}$, the marginal distribution of $X$, $p(x)=[0.4, 0.2, 0.2, 0.2]$, and cost function

\[C=\begin{bmatrix}
	1\ & 2\ & 3\ & 4\\
	\infty\ & 1\ & 2\ & 3\\
	\infty\ & \infty\ & 1\ & 2\\
	\infty\ & \infty\ & \infty\ & 1\\
\end{bmatrix}\]

\noindent the ML-optimal and MI-optimal solutions for 0.5 units of cost are given by:
\[
	P^*_{ML}=
	\begin{bmatrix}
		.25\ & .75\ & 0\ & 0\\
		0\ & 1\ & 0\ & 0\\
		0\ & 0\ & 0\ & 1\\
		0\ & 0\ & 0\ & 1\\
	\end{bmatrix}\quad
	P^*_{MI}=
	\begin{bmatrix}
		.5235\ & .3031\ & 0.1233\ & 0.0502\\
		0\ & .4890\ & .3120\ & .1990\\
		0\ & 0\ & .6105\ & .3895\\
		0\ & 0\ & 0\ & 1\\
	\end{bmatrix}
\]

Finally, since channel capacity is itself defined as a maximization over mutual information, one should expect the same kind of behavior when optimizing over channel capacity.

%Theorem statement and proof (DONE) and Implications
\section{A Heuristic Algorithm} 
In this section, we will address the dimensionality of the LP. For alphabet $|\mathcal{X}|=N$ and $|\mathcal{Y}|=M$, the constrained optimization in Equation~\ref{optProbML} is over an $N\times M$ variable matrix. So, the alphabet sizes of $X$ and $Y$ are intimately linked to the dimensionality of the LP and can greatly affect computational complexity. It may be possible to reduce the problem size by grouping symbols in $X$ or in $Y$ together, thereby reducing $N$ and $M$. However, doing so incurs additional cost by some hard-to-measure quantity and is not always practical. For such cases, we present a heuristic algorithm that can be used to approximate the full trade-off curve.

\subsection{Greedy Algorithm}
\begin{definition}
	For any nonempty set $\mathcal{S}\subseteq\mathcal{Y}$ and cost matrix $\{c(x,y)\}$, we define a deterministic protection scheme $\textbf{P}_\mathcal{S} = \{p_{xy}\}$ such that:
	\begin{equation}
		p_{xy}= \begin{cases}
			1 \quad \text{if } y= \min{\underset{y^\prime\in\mathcal{S}}{\arg\min{}}{c(x,y^\prime)}}\\
			0 \quad \text{otherwise}\\
		\end{cases}
	\end{equation}
	We refer to $\textbf{P}_\mathcal{S}$ as the deterministic protection scheme \emph{induced by} the subset $\mathcal{S}$.
\end{definition}

\begin{definition}
	For any non-empty set $\mathcal{S}\subseteq\mathcal{Y}$, let:
	\begin{equation}
		\mathscr{L}(\mathcal{S})=\mathscr{L}(\textbf{P}_\mathcal{S}) \ and\  \mathscr{C}(\mathcal{S})=\mathscr{C}(\textbf{P}_\mathcal{S})
	\end{equation}
\end{definition}

\begin{definition}
	For a given staircase nondecreasing cost matrix $\{c(x,y)\}$, we identify one (not necessarily unique) $y_0\in\mathcal{Y}$ such that:
	\begin{equation}
		y_0 = \underset{y\in\mathcal{Y}}{\arg\min{}}{\mathscr{C}(\{y\})}
	\end{equation}
	Define the subset $\mathcal{Y}^\prime = \mathcal{Y}-\{y_0\}$.
\end{definition}

\begin{definition} \label{setFunc}
	For any set $\mathcal{A}\subseteq\mathcal{Y}^\prime$, we define the set function:
	\begin{equation}
		f(\mathcal{A})=-\mathscr{C}(\mathcal{A}\cup\{y_0\})
	\end{equation}
\end{definition}

\begin{definition}
	Here, we define a greedy algorithm to construct a sequence of deterministic protection schemes as follows:
	\begin{enumerate}
		\item Start with $\mathcal{A}=\{\emptyset\}$.
		\item Choose $y\in\mathcal{Y}^\prime-\mathcal{A}$ such that $f(\mathcal{A}\cup\{y\})$ is maximized over all such choices of $y$. If $\mathcal{Y}^\prime-\mathcal{A}$ is empty or if there does not exist such $y$ that $f(\mathcal{A}\cup\{y\})>f(\mathcal{A})$, terminate this algorithm.
		\item Set $\mathcal{A}=\mathcal{A}\cup \{y\}$.
		\item Go to step 2.
	\end{enumerate}
\end{definition}

\subsection{Bounded Sub-optimality of the Greedy Algorithm} 
Using standard results in combinatorial optimization \cite{Nemhauser1978}, we can obtain bounds on how suboptimal the solutions obtained from the greedy algorithm are. We will first prove some basic facts about the set function $f(\mathcal{A})$ given in Definition~\ref{setFunc}.
\begin{lemma} \label{submod}
	$f(\mathcal{A})$ is submodular.
\end{lemma}
\begin{proof}
	For $\mathcal{A},\mathcal{B}\subseteq\mathcal{Y}^\prime$ such that $\mathcal{A}\cap\mathcal{B}=\{\emptyset\}$,
	\begin{align*}
		f(\mathcal{A}\cup\mathcal{B}) &= -\sum_{x\in\mathcal{X}}{\min_{y\in\mathcal{A}\cup\mathcal{B}\cup\{y_0\}}{p(x)c(x,y)}}\\
		&= -\sum_{x\in\mathcal{X}}{\min_{y\in\mathcal{A}\cup\{y_0\}}{p(x)c(x,y)}} + \sum_{x\in\mathcal{X}}{\min_{y\in\mathcal{A}\cup\{y_0\}}{p(x)c(x,y)}}\\
		&\phantom{====}- \sum_{x\in\mathcal{X}}{\min_{y\in\mathcal{A}\cup\mathcal{B}\cup\{y_0\}}{p(x)c(x,y)}}\\
		&= f(\mathcal{A}) + \sum_{x\in\mathcal{X}}{p(x)}[\min_{y\in\mathcal{A}\cup\{y_0\}}{c(x,y)}-\min_{y\in\mathcal{A}\cup\mathcal{B}\cup\{y_0\}}{c(x,y)}]\\
		&\equiv f(\mathcal{A}) + D(A,B)
	\end{align*}
	Then, for $\mathcal{A}\subseteq\mathcal{Y}^\prime$ and $b,c\in\mathcal{Y}^\prime\backslash\mathcal{A}$,
	\begin{align*}
		&f(\mathcal{A}\cup\{b\}) + f(\mathcal{A}\cup\{c\}) - f(\mathcal{A}\cup\{b,c\}) - f(\mathcal{A})\\
		&= D(\mathcal{A},\{b\}) + D(\mathcal{A},\{c\}) - D(\mathcal{A},\{b,c\})\\
		&= \sum_{x\in\mathcal{X}}{p(x)}[\min_{y\in\mathcal{A}\cup\{y_0\}}{c(x,y)} - \min_{y\in\mathcal{A}\cup\{b,y_0\}}{c(x,y)}\\
		&\phantom{====} -\min_{y\in\mathcal{A}\cup\{c,y_0\}}{c(x,y)}+\min_{y\in\mathcal{A}\cup\{b,c,y_0\}}{c(x,y)}]\\
		&\equiv \sum_{x\in\mathcal{X}}{p(x)}[C_1-C_2-C_3+C_4]\\
		&\geq 0
	\end{align*}
	since $C_4$ is equal to $C_2$ or $C_3$ (or both), and $C_1$ is no smaller than either $C_2$ or $C_3$. Hence, 
	\begin{equation}
		f(\mathcal{A}\cup\{b\}) + f(\mathcal{A}\cup\{c\}) \geq f(\mathcal{A}\cup\{b,c\}) + f(\mathcal{A})
	\end{equation}
	so $f(\mathcal{A})$ is submodular (\cite{schrijver-book}, Thm 44.1).

%%%%%%%%%%%%%%Keep for original writeup	
%	since $\min_{y\in\mathcal{A}\cup\{b,c,y_0\}}{c(x,y)}$ is equal to $\min_{y\in\mathcal{A}\cup\{b,y_0\}}{c(x,y)}$ or $\min_{y\in\mathcal{A}\cup\{c, y_0\}}{c(x,y)}$ (or both), and $\min_{y\in\mathcal{A}\cup\{y_0\}}{c(x,y)}$ is no smaller than either $\min_{y\in\mathcal{A}\cup\{b,y_0\}}{c(x,y)}$ or $\min_{y\in\mathcal{A}\cup\{c,y_0\}}{c(x,y)}$.
\end{proof}

\begin{definition}
	For integer exp-leak bound $L$, let $\mathcal{A}_{g}(L)$ be the set obtained by running the greedy algorithm unil $|\mathcal{A}\cup\{y_0\}|=L$ (for simplicity, assume the greedy algorithm does not terminate prior to this point). 
	
	For integer exp-leak bound $L$, let $\mathcal{A}^*(L)\subseteq\mathcal{Y}^\prime$ be the true optimal set such that $f(\mathcal{A})$ is maximized subject to $|\mathcal{A}\cup\{y_0\}|\leq L$.
\end{definition}

Now, since $f(\mathcal{A})$  is submodular, we can bound the greedy algorithm for all $L\geq 2$ as follows (\cite{Nemhauser1978}, Theorem 4.1):

\begin{equation} \label{greedyIneq}
	\begin{split}
		\frac{f(\mathcal{A}^*(L))-f(\mathcal{A}_g(L))}{f(\mathcal{A}^*(L))-f(\{\emptyset\})} &= \frac{\mathscr{C}(\mathcal{A}_g(L)\cup\{y_0\})-\mathscr{C}(\mathcal{A}^*(L)\cup\{y_0\})}{\mathscr{C}(\{y_0\})-\mathscr{C}(\mathcal{A}^*(L)\cup\{y_0\})}\\ 
		&\leq \Big(\frac{L-2}{L-1}\Big)^{L-1}\leq \frac{1}{e}
	\end{split}
\end{equation}

The greedy algorithm is capable of approximating a full cost-leakage trade-off curve more quickly, compared to running as many as $M$ individual LP optimizations. The difference in computation time increases with the size of $|\mathcal{Y}|$; the greedy algorithm runs on the order of 30 times faster than  the LP on the integer exp-leakage points for our larger experiments, but only on the order of 5 times faster for our smaller experiments.

Moreover, we have shown that the cost of deterministic protection schemes computed by the greedy algorithm is bounded relative to the true optimal protection schemes at the same leakage levels. Finally, a useful side-effect of this bound is that the true optimal scheme does not perform any better than the greedy algorithm after a single iteration (when $L=2$), which follows from Equation \ref{greedyIneq}. As there exist many applications that require close to no leakage and since a single step of the greedy algorithm (computing for $L=2$) merely consists of a $O(N)$ search over the space of $\mathcal{Y}^\prime$, these protection schemes with exp-leak between 1 and 2 can be easily computed since we know from Theorem~\ref{mainThm} that the optimal protection scheme is simply a convex combination of the two deterministic protection schemes.

\subsection{Sub-optimality of the Greedy Algorithm in Practice}
\begin{figure}
	\centering
	\begin{subfigure}[t]{0.45\columnwidth}
		\includegraphics[width=\columnwidth]{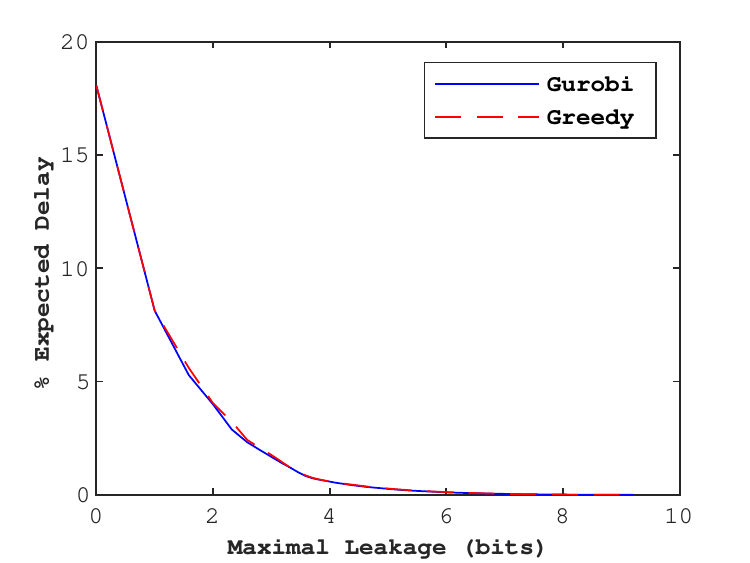}
		\caption{}
		\label{RSAOpt_plot}
	\end{subfigure}
	\begin{subfigure}[t]{0.45\columnwidth}
		\includegraphics[width=\columnwidth]{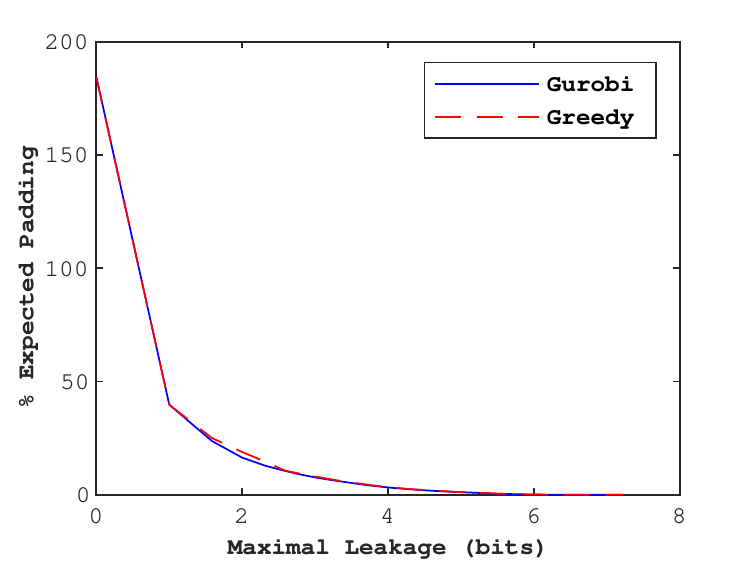}
		\caption{}
		\label{VoIPOpt_plot}
	\end{subfigure}
	\caption{(a) Optimal trade-off curve for 16-bit GMP RSA. The horizontal axis is the leakage bound in bits(the log of $L$ from Equation~\ref{optProbML}) and the vertical axis is the percent expected delay normalized over the baseline expected decryption time (5,246.3 cycles). (b) Analogous trade-off curve for VoIP packet sizes. Here, the vertical axis is the percent expected padding normalized over the baseline expected packet size (54.15 bytes).}
\end{figure}

Here, we reuse the RSA decryption timing data and VoIP packet size data to demonstrate that the gap between the true optimal curve and the greedy algorithm is in fact very small. Indeed, we find that for these case studies, the error of the greedy rate is far below the projected error given by Equation~\ref{greedyIneq}. Thanks to Theorem \ref{mainThm}, we obtain the true optimal curve by using Gurobi to optimize over maximal leakage at integer exp-leak points. We use the greedy algorithm to obtain an approximately optimal curve. These results can be seen in Figures~\ref{RSAOpt_plot} and \ref{VoIPOpt_plot}. %0. Explain necessity 1. Propose algorithm (might as well frame as set function) 2. Bound with Nemhauser result 3. Near-endpoint optimality

%I'm leaving the file in the directory for now, but the idea is to not have the caseStudy section anymore, after absorbing its content into other sections.
%\input{caseStudy}

% This section isn't done yet, so I'm not including it yet.
%\input{LDP}

\section{Related Work} 

There exist many previous studies that demonstrate side channel attacks in various systems and
applications. There also exists work on legitimizing various metrics other than maximal leakage for quantifying leakage in side channels \cite{braun_quantitative_2009,alvim_measuring_2012,smith_foundations_2009,alvim_additive_2014}. Previous work has also proposed many protection schemes against side channel 
attacks, often based on heuristics and without quantitative security guarantees.
Here, we discuss previous work on quantitative metrics for side channels, and studies on
optimal trade-offs between security and protection overhead. 
This paper represents the first to experimentally demonstrate the limitations of traditional
information theoretic metrics, the practical advantages of maximal leakage, and how the optimal 
protection trade-off can be efficiently obtained.

Previous work has attempted to quantify RSA timing channels using conditional entropy (or equivalently, mutual information) with the justification that it quantifies the amount of uncertainty of the adversary's guesses  \cite{kopf_information-theoretic_2007}. In the same vein, useful results on optimal trade-offs using conditional entropy have been developed for strictly deterministic side channels \cite{kopf_provably_2009}. However, in our work, we have made a case that such metrics not only underestimate the threat of the side channel, but past results from using conditional entropy only apply to deterministic side channels, a special case that is subsumed under our models given in Section 2.

Work on developing trade-offs for using maximum entropy has been applied to stochastic side channels \cite{askarov_predictive_2010,zhang_predictive_2011}. Unfortunately, due to the simplicity of maximum entropy (using our notation, maximum entropy is simply $\log{|\mathcal{Y}|}$), protection schemes from this work are restricted to deterministic discretization of the output. While discretization leads to easily implementable protections, they are typically suboptimal with respect to mult-leakage. Trade-offs have also been developed for stochastic side channels using mutual information \cite{mao_quantitative_2017}, but likely due to the nature of mutual information, useful theoretical results are difficult to prove. This previous work proposes discretization and randomization as two possible protection mechanisms. Again, discretization and randomization based on independent noise are suboptimal. Finally some previous work has used mutual information, channel capacity, and maximal leakage all together to provide bounds and trade-offs of each \cite{wang_secure_2017}, but did not try to optimize for them.

Finally, there exists some work in the domain of privacy-preserving publication that analyzes optimal trade-offs using maximal leakage \cite{liao_hypothesis_2017,liao_privacy_2018}. This work is similar in spirit to ours, but the underlying problems are fundamentally different. These results are not for side channels, but for privacy-preservation in publishing database entries. The types of total cost considered are probabilities of Type-II errors and hard distortion functions (which aims to provide strong performance guarantees). In databases, where secrets are typically a small number of bits at most, these results are extremely relevant, but they are hard to justify in most side channels. Moreover, the results attained in this previous work pertain to the use of a tunable form of maximal leakage \cite{liao_tunable_2018} that is largely irrelevant to side channels.

\bibliographystyle{splncs04}
\bibliography{paper}

\appendix
\section{Proofs} \label{appProofs}
%\subsection{Proof of Convexity} \label{appCvx}
%\begin{proof}[Proof of Theorem \ref{ConvexityThm}]
%
%Fix $L_1$, $L_2$, and $\lambda\in [0,1]$ and let $\bar{L}=\lambda L_1 +(1-\lambda)L_2$\\
%Suppose $\textbf{P}_i=\{p_i(y|x)\}$ minimizes $\mathscr{C}(\textbf{P}_i)$ such that $\mathscr{L}(\textbf{P}_i)\leq L_i$ for $i=1,2$. In other words, $\textbf{P}_i$ is optimizing.\\
%Note $\bar{\textbf{P}}=\lambda \textbf{P}_1 + (1-\lambda) \textbf{P}_2$ is a valid transition matrix and \\
%	\begin{align*}
%		\mathscr{L}(\bar{\textbf{P}}) &=  \sum_y{\max_x{\bar{p}(y|x)}}\\
%		&\leq  \lambda \sum_y{\max_x{p_1(y|x)}} + (1-\lambda)\sum_y{\max_x{p_2(y|x)}}\\
%		&= \lambda \mathscr{L}(\textbf{P}_1) + (1-\lambda)\mathscr{L}(\textbf{P}_2)\\
%		&\leq \lambda L_1 + (1-\lambda)L_2\\
%		&= \bar{L}
%	\end{align*}
%
%Thus $\{\bar{p}(y|x)\}$ satisfies the exp-leak bound condition $\mathscr{L}(\bar{\textbf{P}}) \leq \bar{L}$. So by definition,
%
%	\begin{align*}
%		C^*(\bar{L})&\leq \mathscr{C}(\bar{\textbf{P}})\\
%		&= \mathscr{C}(\lambda \textbf{P}_1+ (1-\lambda)\textbf{P}_2)\\
%		&= \sum_x p(x) \sum_y [\lambda p_1(y|x)+(1-\lambda)p_2(y|x)]c(x,y)\\
%		&= \lambda \mathscr{C}(\textbf{P}_1) + (1-\lambda)\mathscr{C}(\textbf{P}_2)\\
%		&= \lambda C^*(L_1)+(1-\lambda)C^*(L_2)\\
%	\end{align*}
%
%Hence $C^*(L)$ is a convex function of $L$.
%\end{proof}

%%%%%%OLD TITLE%%%%%%%%\subsection{Proof of Key Lemma} \label{appLem}
\subsection{Proof of Theorem \ref{mainThm}.1} \label{appLem}
%%%%%%%%%%%%%%%%%%%%%%%%%%%%WATER FILLING%%%%%%%%%%%%%%%%%%%%%%%%%%%%%%%%%%%%
First, we establish a structural result for all optimal protection schemes that will be a useful assumption for later steps of our proof of Theorem \ref{mainThm}.1
\begin{definition} \label{waterFillingDef}
% "water-filled" matrix and "water-filled" construction
\textbf{(Water-Filling)} 

Consider any protection scheme \textbf{P}. Define a $1\times N$ vector $\vec{p}= [p_1, p_2,...p_N]$ such that $p_i = \max_{x\in\mathcal{X}}{p_{xy_i}}$. In plain terms, $\vec{p}$ consists of the column maxima of \textbf{P}.
Using $\vec{p}$ alone, we construct a new protection scheme \textbf{P'} as follows:
	\begin{enumerate}
		\item Start with a $M\times N$ zero matrix $\textbf{P'}=\{p_{xy}^\prime\}$. We will assume that the members of $\mathcal{X}$ and $\mathcal{Y}$ are in some enumerated order, as previously stipulated.
		\item For each row $x\in\mathcal{X}$, iterate over each $y_i,\ i = 1,2,...N$.
			\begin{itemize}
				\item If $c(x,y_i)=\infty$, let $p_{xy_i}^\prime = 0$
				\item Else, set $p_{xy_i}^\prime = \min{\{p_i,1-\sum_{j=1}^{i-1}{p_{xy_j}^\prime}\}}$.
			\end{itemize}
	\end{enumerate}
In plain terms, we are constructing \textbf{P'} by maintaining the column maxima of \textbf{P} and "filling" in probability mass in each row from left to right.

We define this procedure to generate \textbf{P'} from \textbf{P} as the method to convert \textbf{P} into "water-filled" form. Also, if \textbf{P} and \textbf{P'} are identical, we say that \textbf{P} is a "water-filled" protection scheme.
\end{definition}

\begin{lemma} \label{waterFillingLemma}
% All optimizing P can be re-written in "water-filled" form and still achieve the same total cost and exp-leak
\textbf{(Water-Filling Lemma)} 

If the cost function satisfies definition \ref{stairNondec}, then all optimizing \textbf{P} can be converted into water-filled form \textbf{P'} such that $\mathscr{C}(\textbf{P})=\mathscr{C}(\textbf{P'})$ and $\mathscr{L}(\textbf{P})=\mathscr{L}(\textbf{P'})$.
\end{lemma}

\begin{proof}[Lemma \ref{waterFillingLemma}]
Suppose we are given optimizing \textbf{P} and its water-filled form \textbf{P'}. By its construction, $\mathscr{L}(\textbf{P})\geq\mathscr{L}(\textbf{P'})$ since we did not increase the total sum of column maxima. In addition, since we independently fill up each row's entries in \textbf{P'} from least cost to greatest cost, $\mathscr{C}(\textbf{P})\geq\mathscr{C}(\textbf{P'})$ for any cost function that is staircase nondecreasing.

From these two statements, we also obtain the reverse inequalities:
	\begin{itemize}
		\item $\mathscr{C}(\textbf{P})\geq\mathscr{C}(\textbf{P'})$ implies that $\mathscr{L}(\textbf{P})\leq\mathscr{L}(\textbf{P'})$ since \textbf{P} is optimizing and \textbf{P'} cannot perform any better (have lower $\mathscr{C}+\alpha\mathscr{L}$) than \textbf{P}.
		\item Similarly, $\mathscr{L}(\textbf{P})\geq\mathscr{L}(\textbf{P'})$ implies that $\mathscr{C}(\textbf{P})\leq\mathscr{C}(\textbf{P'})$, since \textbf{P} is optimizing.
	\end{itemize}

Therefore, $\mathscr{C}(\textbf{P})=\mathscr{C}(\textbf{P'})$ and $\mathscr{L}(\textbf{P})=\mathscr{L}(\textbf{P'})$.
\end{proof}

%%%%%%%%%%%%%%%%%%%%%%%%BEGIN MAIN PROOF%%%%%%%%%%%%%%%%%%%%%%%%%%%%%%%%%%%%%
\begin{remark} \label{pPrimeArg}
% Explanation of approach.
%%%%%OLD TITLE%%%%%%\textbf{(Proof Approach for Lemma \ref{keyLemma})} 
\textbf{(Proof Approach for Theorem \ref{mainThm}.1)}
For any optimizing \textbf{P}, we start by assuming it is already in water-filled form, since we have already shown that doing so does not unnecessarily restrict our space of optimizing solutions. Then, we would like to show that there exists a special choice of \textbf{Q} such that  $\mathscr{C}(\textbf{P}+\delta\textbf{Q}) + \alpha \mathscr{L}(\textbf{P}+\delta\textbf{Q})$: 
	\begin{enumerate}
		\item is linear over some well-defined interval of $\delta$ values around 0.
		\item does not change with $\delta$ for any fixed $\alpha$
		\item results in protection scheme $\textbf{P}+\delta\textbf{Q}$ being strictly "more deterministic" (to be defined shortly) than \textbf{P} for a particular choice of $\delta$.
	\end{enumerate}
\end{remark}

\begin{definition}
\textbf{(Types of Matrix Entries)}

% fractional entries, maxed out entries, and hanging entries
For the sake of discourse, we will define the following classifications of matrix entries in any protection scheme:
	\begin{itemize}
		\item An entry is \emph{fractional} if it is not equal to 0 or 1, and \emph{integral} otherwise. Similarly, a column is fractional if its maximum entry is fractional and integral otherwise.
		\item An entry is \emph{maxed out} if it is equal to the maximum value in its column, and \emph{hanging} otherwise.
	\end{itemize}
\end{definition}

\begin{remark} \label{WF-prop}
It is true by construction that a water-filled protection scheme will have at most one hanging mass entry and at least one maxed out entry in each row. Moreover, if a row has a hanging mass entry, there do not exist other non-zero entries further to the right of that entry.
\end{remark}

\begin{remark}
% R(P) definition
\textbf{(Measure of Randomness)}

In order to compare which protection scheme, between two options, is "more deterministic", we rely on the following metric for randomness of a protection scheme:

$R(\textbf{P}) = $ (\# fractional columns in \textbf{P}) + (\# hanging entries in \textbf{P})

Note that $R(\textbf{P})=0$ if and only if \textbf{P} is a deterministic protection scheme.
\end{remark}

%%%%%%%%%%%%%%%%%%%%%%%%%%%%P' CONSTRUCTION%%%%%%%%%%%%%%%%%%%%%%%%%%%%%%%%%%
We will now propose a particular choice of \textbf{Q} and $\delta$, and prove the desired properties about these choices after.

\begin{definition} \label{QGen}
\textbf{(Q-Generation Procedure)}

Given the water-filled protection scheme \textbf{P} with at least one fractional entry, we now define a procedure to generate a \textbf{Q} matrix. Note that any such protection scheme must also have at least one fractional column or else it would contradict the water-filled property.
	\begin{enumerate}
		\item Start with an $M\times N$ zero matrix \textbf{Q} that we will populate with values.
		\item Denote the leftmost fractional column index in \textbf{P} as $y$. Further denote the current "sign" to "+". 
		\item In the $y$th column of \textbf{Q}, if the sign is "+", assign the value $1$ to all entries in that column that are maxed out in \textbf{P}. If the sign is "-", assign the value $-1$ instead. 
		\item If the current sign is "+", change it to "-", and vice versa.
		\item Consider the set of rows that are maxed out in the $y$th column of \textbf{P}. \textit{Do all of these rows either have hanging mass in \textbf{P} or already have 2 non-zero entries in \textbf{Q}?} Depending on the answer:
			\begin{itemize}
				\item If yes, go to step 9.
				\item If no, then proceed to step 6.
			\end{itemize}
		\item Again consider the set of rows that are maxed out in the $y$th column of \textbf{P}. Choose the topmost row from this set that does not have hanging mass in \textbf{P} and has only 1 non-zero entry in \textbf{Q}. Denote the row index of that entry as $x$.
		\item Set $y$ to be the column index of the rightmost, maxed out entry of the $x$th row in the \textbf{P} matrix. Note that $y$ must correspond to a fractional column here.
		\item Go to step 3.
		\item If any rows in \textbf{Q} have hanging mass and an odd number of non-zero entries, assign either $1$ or $-1$, so that each of these rows sum to 0, to the hanging mass entries of these rows.
	\end{enumerate}
\end{definition}

\begin{lemma} \label{validQ}
\textbf{(Q-Generation Properties)} 
The procedure specified by definition \ref{QGen} satisfies the following:
	\begin{enumerate}
		\item The procedure terminates.
		\item All of the rows in the generated \textbf{Q} matrix sum to 0 (so that $\textbf{P}+\delta\textbf{Q}$ is a protection scheme).
		\item $\textbf{P}+\delta\textbf{Q}$ is a water-filled protection scheme
	\end{enumerate}
\end{lemma}

\begin{proof}[Lemma \ref{validQ}.1]
Since we never choose columns that aren't fractional, any row selected in step 6 must have a maxed out entry (because we also ignore rows with fractional entries) somewhere to the right of the current $y$ column. Certainly, this procedure must terminate if the $y$ value ever reaches the right-most column (and the process may terminate earlier than that due to step 5).
\end{proof}

\begin{proof}[Lemma \ref{validQ}.2]
Since we only assign $1$ and $-1$ to entries of \textbf{Q} in alternation, this is the same as saying that each row must contain an even number of non-zero entries. We see that this is true by noting that there are three types of rows, differentiated by how their non-zero entries in \textbf{Q} (if any) are assigned during the Q-generating procedure.

If a row has hanging mass in \textbf{P}, then step 9 will necessarily adjust that row to have an even number of non-zero entries by construction. In addition, we never assign mass to hanging mass entries until step 9, when the procedure terminates, which means that all hanging mass entries are free for us to use at that point. So, rows that have hanging mass in \textbf{P} will be valid rows in \textbf{Q}.

If a row has no hanging mass in \textbf{P}, then there are two cases, depending on whether that row was ever used in step 6 to determine the next $y$ value (we'll refer to such a row as a "critical" one). Note that, due to steps 5 and 6 filtering out rows that already have 2 non-zero entries, no row will ever be used in step 6 twice (i.e. a row will be a critical row at most once).
	\begin{itemize}
		\item If the row is critical, it must be the topmost one that had only one non-zero entry in \textbf{Q} at that point of the procedure in the previous $y$th column. Step 7 guarantees that the only other non-zero entry in this row will correspond to its rightmost non-zero entry in \textbf{P}. So this row will have exactly 2 non-zero entries in \textbf{Q}, making it valid.
		\item If the row is not critical, it must either be located below one that is or not have any non-zero entries in \textbf{Q} at all.  The latter case results in a trivially valid row. In the former case, the row must have at least two non-zero entries in columns shared with the previous critical row, or else it would violate our assumptions that \textbf{P} is water-filled and the cost function is staircase nondecreasing. In addition, since \textbf{P} is water-filled, each row is majorized by all rows above it (i.e. the cumulative left-to-right sum of the upper row is no less than that of the lower row for every column). This implies that a non-critical row cannot have more than 2 non-zero entries either. 
	\end{itemize}
\end{proof}

\begin{proof}[Lemma \ref{validQ}.3]
We observe that due to step 3, we only ever change all of the maxed out entries in a column together. So, for small $\delta$, $\textbf{P}+\delta\textbf{Q}$ will remain water-filled.
\end{proof}

\begin{definition} \label{stop}
% stopping conditions
\textbf{(Stopping Conditions)}

Recall from remark \ref{pPrimeArg} that we require a particular choice of $\delta$ with various properties, as already described. We now define two choices of $\delta$ and justify properties about them in later lemmas.

Let $\delta_+ = \sup[\delta\geq 0: \textbf{P}+\delta\textbf{Q}$  is stochastic and \textbf{P} and $\textbf{P}+\delta\textbf{Q}$  are maxed out for the same entries and fractional for the same entries]

and $\delta_- = \inf[\delta\leq 0: \textbf{P}+\delta\textbf{Q}$  is stochastic and \textbf{P} and $\textbf{P}+\delta\textbf{Q}$  are maxed out for the same entries and fractional for the same entries]

Note that, by definition $\delta_+>0$ and $\delta_-<0$.
\end{definition}

%%%%%%%%%%%%%%%%%%%%%%%%%%%%%%%%LEMMAS%%%%%%%%%%%%%%%%%%%%%%%%%%%%%%%%%%%%%%%
\begin{lemma} \label{LinearityLemma}
% Linearity
\textbf{(Linearity Lemma)}

If \textbf{P} is water-filled for fixed $\alpha$ and \textbf{Q} is generated according to definition \ref{QGen}, then $\mathscr{C}+\alpha\mathscr{L}$ evaluated with $\textbf{P}+\delta\textbf{Q}$ is linear with respect to $\delta\in [\delta_-,\delta_+]$.
\end{lemma}
\begin{proof}[Lemma \ref{LinearityLemma}]
For $\delta_-<\delta<\delta_+$ and fixed $\alpha$,

\begin{align*}
	& \mathscr{C}(\textbf{P}+\delta\textbf{Q}) + \alpha\mathscr{L}(\textbf{P}+\delta\textbf{Q})\\
	&= \underset{x}{\sum}\underset{y}{\sum} p(x)c(x,y)(p_{xy}+\delta q_{xy}) + \alpha \underset{y}{\sum}\underset{x}{\max}(p_{xy}+\alpha q_{xy})\\
	&= \underset{x}{\sum}\underset{y}{\sum} p(x)c(x,y)(p_{xy}+\delta q_{xy}) + \alpha \underset{y}{\sum}(p_{x(y)y}+\alpha q_{x(y)y})
\end{align*}

where $x(y)=\underset{x}{\arg\max}\ p_{xy}$.

Since $\mathscr{C}(\textbf{P}+\delta\textbf{Q}) + \alpha\mathscr{L}(\textbf{P}+\delta\textbf{Q})$ is linear over $(\delta_-,\delta_+)$ and continuous over $[\delta_-,\delta_+]$, it is linear over $[\delta_-,\delta_+]$.
\end{proof}

\begin{lemma} \label{NoImprLemma}
% Does not change C+aL
\textbf{(No Improvement Lemma)}

If \textbf{P} minimizes $\mathscr{C}+\alpha\mathscr{L}$ over $S$ for fixed $\alpha$ and is water-filled and \textbf{Q} is generated according to definition \ref{QGen}, then $\frac{\partial}{\partial\delta} (\mathscr{C}+\alpha\mathscr{L}) = 0$ at $\delta=0$.
\end{lemma}
\begin{proof}[Lemma \ref{NoImprLemma}]
If $\frac{\partial}{\partial\delta} (\mathscr{C}+\alpha\mathscr{L}) \neq 0$, then that implies that $\textbf{P}+\delta\textbf{Q}$ performs strictly better for some $\delta$ close to zero, which is a contradiction.
\end{proof}

\begin{lemma} \label{MoreDetLemma}
% Becomes more deterministic
\textbf{(More Deterministic Lemma)}

If \textbf{P} is water-filled for fixed $\alpha$ and \textbf{Q} is generated according to Definition \ref{QGen}, then 
$R(\textbf{P}+\delta\textbf{Q})<R(\textbf{P})$ for both $\delta=\delta_-$ or $\delta=\delta_+$ as defined by Definition \ref{stop}.
\end{lemma}
\begin{proof}[Lemma \ref{MoreDetLemma}]
As $\delta$ increases from 0 to $\delta_+$, some fractional entries of $\textbf{P}+\delta\textbf{Q}$ change, and none of the integral entries change. In addition, if one maxed out entry changes, all of the maxed out entries in that column change together. It thus follows that the set of fractional columns can only decrease with $\delta$ and that the set of hanging entries likewise can only decrease. So $R(\textbf{P})$ is nonincreasing in $\delta$ for $\delta\in[0,\delta_+]$. From the definition of $\delta_+$ in Definition \ref{stop}, $R(\textbf{P}+\delta_+\textbf{Q})<R(\textbf{P})$.

Similarly, we can show that $R(\textbf{P}+\delta_-\textbf{Q})<R(\textbf{P})$.
\end{proof}

\begin{proof}[Theorem \ref{mainThm}.1]
Any \textbf{P} that minimizes $\mathscr{C}+\alpha\mathscr{L}$ for  some $\alpha$ can be chosen to be optimizing and water-filled as per Lemma \ref{waterFillingLemma}. If \textbf{P} is not a deterministic protection scheme, we can select \textbf{Q} as in Definition \ref{QGen} with the properties shown in Lemma \ref{validQ}. 

By Lemmas \ref{LinearityLemma}, \ref{NoImprLemma}, \ref{MoreDetLemma}, we know $\mathscr{C}(\textbf{P}+\delta\textbf{Q})+\alpha\mathscr{L}(\textbf{P}+\delta\textbf{Q})$ is constant over $[\delta_-,\delta_+]$ and $R(\textbf{P}+\delta\textbf{Q})<R(\textbf{P})$.

If $\textbf{P}+\delta\textbf{Q}$ is not deterministic, then we can repeat the above process since it is still water-filled and minimizes $\mathscr{C}+\alpha\mathscr{L}$ for the same $\alpha$.

Eventually, after repeating this process some finite number of times, $R(\textbf{P})$ will be 0 (since the function we defined is always nonnegative), and therefore deterministic.
\end{proof}

\subsection{Proof of Theorem \ref{mainThm}.2} \label{appThm}
\begin{proof}[Theorem \ref{mainThm}.2]
Using standard convex analysis (e.g. \cite{Rockafellar}, chapter 12), Theorem \ref{mainThm}.1 implies that that $C^*(L)$ and $C^*_d(L)$ have the same lower semi-continuous hull (or the closure, as defined by \cite{Rockafellar} chapter 7), which is equivalent to our definition of the boundary of $S$. We can see this fact as follows:

First, we note that the left and right hand sides of the equality in Theorem \ref{mainThm}.1 are the conjugate functions of $C^*(L)$ and $C^*_d(L)$, respectively. We have shown that the conjugates are equal for any $\alpha$.

Second, since $C^*(L)$ is a convex function of $L$, the conjugate of the conjugate of $C^*(L)$ is equal to the closure of $C^*(L)$ (\cite{Rockafellar}, Corollary 13.1.1).

Third, while $C^*_d(L)$ is not a convex function, its conjugate is the same as the conjugate of the closure of its convex hull. Therefore, the conjugate of its conjugate must be equal to the closure of its convex hull.

Thus, we have shown that the convex hulls of $S$ and $S_d$ are the same, since the two sets are the epigraphs of (all points in $\mathbb{R}^2$ on or above the curves defined by) the functions $C^*(L)$ and $C^*_d(L)$, respectively. 

From this fact, it trivially follows that any $(L,C)$ pair on the boundary of $S$ must also lie on the convex hull of $S$, and therefore on the convex hull of $S_d$.

Finally, as previously noted, $C^*_d(L)$ is a descending staircase-like function for $L\in [1,\infty]$. So, the convex hull of $S_d$ is given by the largest convex linear interpolation of the outer corner points of $C^*_d(L)$ (for example, see figure)

Therefore, any $(L,C)$ pair on the boundary of $S$ is achievable by a convex combination of no more than two deterministic protection schemes.
\end{proof}

\end{document}